\documentclass{article}

\usepackage{microtype}
\usepackage{graphicx}
\usepackage{subcaption}
\usepackage{booktabs} % for professional tables
\usepackage{algorithm2e}
\usepackage{lipsum}
\usepackage{float}

\usepackage{hyperref}

\usepackage[accepted]{icml2026_testing}

\usepackage{amsmath}
\usepackage{amssymb}
\usepackage{mathtools}
\usepackage{amsthm}
\usepackage{bbm}

\usepackage[capitalize,noabbrev]{cleveref}

\theoremstyle{plain}
\newtheorem{theorem}{Theorem}[section]
\newtheorem{proposition}[theorem]{Proposition}

\newtheorem{corollary}[theorem]{Corollary}
\theoremstyle{definition}
\newtheorem{definition}[theorem]{Definition}

\theoremstyle{plain}
\newtheorem{remark}[theorem]{Remark}

\usepackage[textsize=tiny]{todonotes}

\icmltitlerunning{\textbf{Scalable Approximate $\text{SNR}$-Optimised Polynomial Stein Discrepancies}}

\begin{document}

\twocolumn[
  \icmltitle{$\lambda$\textbf{-PSD: Scalable Approximate $\text{SNR}$-Optimised Polynomial Stein Discrepancies}}



  \icmlsetsymbol{equal}{*}

\begin{icmlauthorlist}
    \icmlauthor{Minh-Long Nguyen}{xxx,zzz}
    \icmlauthor{Thanh-Long Vu}{xxx,comp}
    \icmlauthor{Christopher Drovandi}{xxx,yyy,zzz}
    \icmlauthor{Leah F. South}{xxx,zzz}
    \icmlauthor{Trung-Tin Nguyen}{xxx,yyy,zzz}
\end{icmlauthorlist}

\icmlaffiliation{xxx}{School of Mathematical Sciences, Queensland University of Technology, Australia}
\icmlaffiliation{comp}{Quantium, Brisbane, Australia}
\icmlaffiliation{yyy}{ARC Centre of Excellence for the Mathematical Analysis of Cellular Systems (MACSYS)}
\icmlaffiliation{zzz}{Centre for Data Science, Queensland University of Technology, Australia}

  \icmlcorrespondingauthor{Minh-Long Nguyen}{minhlong.nguyen@hdr.qut.edu.au}

  \icmlkeywords{Machine Learning, ICML}

  \vskip 0.3in
]



\printAffiliationsAndNotice{}  % no special notice (required even if empty)

    \begin{abstract}
Polynomial Stein discrepancies (PSD) provide a scalable alternative to kernel Stein methods for measuring sample quality and goodness-of-fit testing, but their statistical properties remain poorly understood. We show that increasing polynomial degree primarily amplifies signal without adequately controlling variance, rather than directly optimising the signal-to-noise ratio (SNR). Under suitable assumptions, this might lead to a failure mode in which the $\text{SNR}^2$ can provably decay exponentially with polynomial degree. Motivated by this observation, we reformulate Stein discrepancy construction as an explicit $\text{SNR}^2$ maximisation problem, yielding a Rayleigh quotient over Stein features. This perspective motivates $\lambda$-PSD, an approximate scalable covariance-aware reweighting scheme defined in a low-dimensional subspace. Under Gaussian settings, we show that $\lambda$-PSD avoids the exponential $\text{SNR}^2$ collapse and achieves a stable $\text{SNR}^2$. Empirically, $\lambda$-PSD substantially improves test power while retaining linear-time complexity in the number of samples, highlighting the importance of SNR-aware design for scalable Stein discrepancies.
\end{abstract}

{\it Keywords:} Computational statistics $\cdot$ Stein's method $\cdot$ sample quality measure  $\cdot$ hypothesis testing  $\cdot$ scalable Monte Carlo

\section{Introduction}

Assessing sampling quality remains a fundamental problem in modern Bayesian computation \citep{bhattacharya2024grand}\nocite{brooks2011handbook, gelman1992inference,vehtari2021rank}. In many settings, the target distribution \(p\) is known only up to a normalising constant, while the approximation \(q\) is represented through samples. Stein discrepancies provide a principled framework for directly measuring discrepancy between \(p\) and \(q\) using only score evaluations of the target distribution \citep{gorham2015measuring}. Among these, kernel Stein discrepancy (KSD; \citealt{chwialkowski2016kernel, liu2016kernelized, gorham2017measuring}) have strong theoretical guarantees but typically require quadratic computational complexity in the number of samples $n$, limiting scalability in large-scale settings \citep{jitkrittum2017linear}.
Recent work on polynomial Stein discrepancy (PSD; \citealt{srinivasan2024polynomial}) addresses this limitation by replacing kernel witness functions with polynomial Stein features of order \(Q\), yielding computationally efficient linear-time discrepancy measures (see Section~\ref{sec:related} for further discussion of linear-time Stein discrepancies). PSD is particularly effective at detecting moment convergence under the Bernstein-von Mises limit \citep{srinivasan2024polynomial}. However, its statistical power is highly sensitive to the choice of polynomial order \(Q\), and the statistical mechanisms underlying this behaviour remain poorly understood. 

\subsection{Contributions}

Our main contributions are summarised as follows:

\begin{itemize}
    
    \item We identify and theoretically characterise a previously unstudied limitation of PSD, showing that the asymptotic scaling of $\text{SNR}^2$ may decay exponentially with polynomial degree under Gaussian settings.

    \item We introduce an $\text{SNR}^2$-optimised Stein discrepancy framework that reformulates discrepancy construction as a Rayleigh quotient optimisation problem, leading to the proposed \(\lambda\)-PSD estimator, which achieves a stable $\text{SNR}^2$ under Gaussian settings; we further discuss extensions and generalisations of the \(\lambda\)-PSD framework.

    \item We demonstrate through simulation studies that \(\lambda\)-PSD substantially improves robustness to the choice of \(Q\) and achieves competitive goodness-of-fit testing performance while retaining linear-time complexity.
    
\end{itemize}

\section{Related Works}
\label{sec:related}
\paragraph{(Kernel) Stein Discrepancies}

Stein discrepancies provide a principled framework for directly measuring discrepancy between \(p\) and \(q\) using only score evaluations of the target distribution \citep{gorham2015measuring}. By combining Stein operators \citep{stein1972bound} with rich function classes, they construct integral probability metrics \cite{muller1997integral} that goes to zero when the sampling distribution matches the target. Stein discrepancies have since been widely applied to goodness-of-fit testing, importance sampling, and model training, to just name a few \citep{liu2016kernelized, liu2016stein, liu2017black, grathwohl2020learning}. 

A large body of work has focused on kernel Stein discrepancy (KSD), which combines the first-order Langevin-Stein operator with reproducing kernel Hilbert spaces to obtain computationally tractable discrepancy measures \citep{liu2016kernelized, chwialkowski2016kernel, gorham2017measuring}. Under suitable kernels and regularity conditions on $p$, KSD is convergence determining, meaning that vanishing discrepancy implies weak convergence of the empirical distribution to the target distribution \citep{gorham2017measuring}.  However, standard KSD estimators require quadratic computational complexity due to pairwise kernel evaluations, limiting scalability in large-sample settings \citep{jitkrittum2017linear}.

\paragraph{Goodness-of-Fit Testing}

In goodness-of-fit testing\footnote{We consider the hypothesis testing problem 
\(H_0: q = p\) against \(H_1: q \neq p\).}, many Stein-based methods, particularly KSD-based tests, rely on asymptotic null distributions or (wild) bootstrap procedures \citep{leucht2013dependent,  chwialkowski2016kernel, liu2016kernelized}. A smaller class of methods instead uses asymptotically normal test statistics, enabling $z$-test-style calibration \citep{grathwohl2020learning, repasky2023neural}. While some Stein-based tests rely on heuristic constructions, such as the widely used median heuristic for Gaussian kernels \citep{garreau2017large}, many approaches explicitly optimise test power through adaptive feature or kernel design \citep{jitkrittum2017linear, grathwohl2020learning}.

\paragraph{Linear-Time Stein Discrepancies.}
Standard KSD estimators typically incur quadratic computational complexity due to pairwise kernel evaluations. This has motivated scalable alternatives such as the finite-set Stein discrepancy (FSSD; \citealt{jitkrittum2017linear}), random-feature Stein discrepancy (RFSD; \citealt{huggins2018random}) and polynomial Stein discrepancy (PSD; \citep{srinivasan2024polynomial}, all of which achieve (near) linear-time complexity through finite-dimensional feature constructions. FSSD relies on carefully tuned test locations and kernel bandwidths, typically requiring sample splitting \citep{jitkrittum2017linear}, while RFSD depends on several example-specific design choices and may perform poorly under indirect sampling \citep{srinivasan2024polynomial}. Unlike most methods discussed here, PSD specifically targets moment discrepancies. Its power need not deteriorate with increasing $d$ when the mismatch is distributed across many coordinates \citep{srinivasan2024polynomial}. However, PSD is sensitive to the choice of $Q$, and the mechanisms underlying this behaviour remain poorly understood.

\section{Polynomial Stein Discrepancy in $\mathcal{L}_1$}

We first introduce some common notation used throughout this section. Let \(p\) denote the target probability measure on $\mathbb{R}^d$ with continuously differentiable density. We assume the score function \(s_p \coloneqq \nabla \log p\) can be evaluated, although direct integration under \(p\) is generally intractable \citep{gorham2017measuring}. Let \(q\) denote another probability measure on \(\mathbb{R}^d\) with continuously differentiable density, known through samples \(\textbf{x} = x_1,\dots,x_n\). We are interested in quantifying the discrepancy between $q$ and $p$ through \textbf{x}. Most results in this section are stated under the assumption that \(x_1,\dots,x_n\) are i.i.d.\ samples; although this assumption can often be relaxed, it does not restrict the methods to i.i.d.\ sampling settings.

We denote by \( \langle \cdot, \cdot \rangle_F \) the Frobenius inner product between two matrices of the same dimension. The corresponding Frobenius norm is denoted by \( \|\cdot\|_F \). For a vector \( v \), \( \mathrm{diag}(v) \) denotes the diagonal matrix with entries of \( v \) on its diagonal. Let \( \mathrm{vec}(\cdot) \) denote the vectorisation operator that stacks the columns of a matrix into a single vector. Finally, let $z_{[i]}$ denote the $i$th component of a vector $z$.

\subsection{Definitions}

We consider the first-order Langevin-Stein operator
\begin{equation}
\mathcal{L}_1 h(x) = s_p(x)^\top h(x) + \nabla \cdot h(x),
\end{equation}
for a vector-valued test function $h:\mathbb{R}^d \to \mathbb{R}^d$. To obtain a computationally efficient class of test functions, we restrict $h$ to a non-interacting polynomial basis
\begin{equation}
\label{eqn:poly}
h(x) = \sum_{i=1}^{Q} A_i x^{\odot i} + b, 
\end{equation}
where $A_i \in \mathbb{R}^{d \times d}$, $b \in \mathbb{R}^d$, and $\odot$ denotes the Hadamard product. This class admits a degenerate case when $Q = 0$, in which the test function reduces to the constant vector field $h(x) = b$. 

It is possible to show that
\begin{equation}
\mathcal{L}_1 h(x) = \sum_{i=1}^{Q} \langle A_i, \Phi_i(x) \rangle_F + b^\top s_p(x),
\end{equation}
where the feature maps $\Phi_i(x) \in \mathbb{R}^{d \times d}$ are given by
\begin{equation}
\Phi_i(x) = s_p(x)(x^{\odot i})^\top + i\,\mathrm{diag}(x^{\odot(i-1)}).
\end{equation}

We define the following polynomial Stein discrepancy, denoted by $\mathrm{PSD}_{\mathcal{L}_1}$, to distinguish it from the $\mathrm{PSD}_{\mathcal{L}_2}$ introduced in \citet{srinivasan2024polynomial}, which is based on the second-order Langevin-Stein operator. Although the constructions differ, analogous generalisations are expected to hold. To the best of our knowledge, neither this $\mathcal{L}_1$ formulation nor its statistical properties have been previously derived or analysed.

\begin{definition}
The polynomial Stein discrepancy (PSD) associated with $\mathcal{L}_1$ is defined as
\begin{equation}
\mathrm{PSD}_{\mathcal{L}_1}(q \,\|\, p) := \sup_{\{A_i\}_{i=1}^Q,\, b} \; \mathbb{E}_{q} \bigl[ \mathcal{L}_1 h(x) \bigr],
\end{equation}
such that
$
\sum_{i=1}^{Q} \|A_i\|_F^2 + \|b\|_2^2 \le 1.
$
\end{definition}
Under regularity condition on tails of $p$ (see \citet{south2022semi, south2022postprocessing}) \footnote{that we assume to hold through-out this work}, if $q = p$, then $\mathbb{E}_p[\mathrm{PSD}_{\mathcal{L}_1}(q \,\|\, p)] = 0$; hence, deviations from zero may indicate a discrepancy between \(p\) and \(q\). 

The maximiser of $\mathrm{PSD}_{\mathcal{L}_1}(q \,\|\, p)$ admits a closed-form expression:

\begin{proposition} \label{prop:optimal}
Let $\mu_i = \mathbb{E}_q[\Phi_i(x)] \in \mathbb{R}^{d\times d}$ for $i \in \{1, \dots, Q\}$, and $\mu_b = \mathbb{E}_q[s_p(x)] \in \mathbb{R}^{d}$. The maximiser of $\mathrm{PSD}_{\mathcal{L}_1}(q \,\|\, p)$ is given by
\begin{equation}
A_i^\star = \frac{\mu_i}{Z}, \qquad b^\star = \frac{\mu_b}{Z},
\end{equation}
where the normalisation constant $Z$ is also the optimal discrepancy value:
\begin{equation}
\label{eqn:optimal}
\mathrm{PSD}_{\mathcal{L}_1}(q \,\|\, p) = Z = \sqrt{\sum_{k=1}^{Q}\|\mu_k\|_F^2 + \|\mu_b\|_2^2}.
\end{equation}
\end{proposition}

In practice, we use the empirical  $\widehat{\mathrm{PSD}}_{\mathcal{L}_1}(q \,\|\, p)$, obtained by replacing the population moments $\{\mu_i, \mu_b\}$ with their respective sample estimates $\{\hat{\mu}_i, \hat{\mu}_b\}$.

In the degenerate case $Q = 0$, the optimization reduces to $\sup_{\|b\|_2 \le 1} b^\top \mu_b$, and the empirical discrepancy simplifies to the $\ell_2$-norm of the empirical average score:
\begin{equation}
\widehat{\mathrm{PSD}}_{\mathcal{L}_1}(q \,\|\, p) = \left\|\hat{\mu}_b\right\|_2 = \left\| \frac{1}{n}\sum_{j=1}^n s_p(x_j) \right\|_2.
\end{equation}

\subsection{Time Complexity}
The empirical Stein moments $\hat{\mu}_i = \frac{1}{n} \sum_{j=1}^n \Phi_i(x_j)$ and $\hat{\mu}_b = \frac{1}{n} \sum_{j=1}^n s_p(x_j)$ allow for a closed-form evaluation of the discrepancy. Specifically, the total computational complexity to evaluate $\widehat{\mathrm{PSD}}_{\mathcal{L}_1}$ is $\mathcal{O}(Qnd^2)$. 
While we focus on this formulation, a scalable $\mathcal{O}(Qnd)$ variant can be obtained by restricting $A_i$ to be diagonal.

\subsection{Characterisation under the Bernstein-von-Mises limit}

Similar to \(\mathrm{PSD}_{\mathcal{L}_2}\), the proposed \(\mathrm{PSD}_{\mathcal{L}_1}\) captures moment convergence under the Bernstein-von Mises limit. However, due to its lower-parameter construction, it characterises at most \(Qd^2+d\) moment conditions, compared to the \(O\!\left(\binom{Q+d}{d}\right)\) conditions captured by \(\mathrm{PSD}_{\mathcal{L}_2}\) \citep{srinivasan2024polynomial}. 

\begin{theorem}\label{thm:BvM_PSD_degreeQ_general}
Let \(p=\mathcal{N}(\mu_p,\Sigma_p)\) with positive-definite $\Sigma_p$, so
\(s_p(x)=-\Sigma^{-1}(x-\mu_p)\). For \(Q\ge1\), define
\begin{equation}
h(x)=\sum_{i=1}^Q A_i\big(x^{\circ i}\big)+b,
\qquad A_i\in\mathbb{R}^{d\times d},\; b\in\mathbb{R}^d.
\end{equation}
Then for any \(q\) with finite moments up to order \(Q+1\),
\begin{equation}
\mathrm{PSD}_{\mathcal{L}_1}(q\,\|\,p)=0
\;\Longleftrightarrow\;
\left\{
\begin{aligned}
&\mathbb{E}_q[X] = \mu_p,\\
&\mathbb{E}_q\!\left[(X_{[l]}-\mu_{p[l]})\,X_{[j]}^{i}\right] 
= \\ &i\,\Sigma_{lj}\,\mathbb{E}_q\!\left[X_{[j]}^{i-1}\right],
\end{aligned}
\right.
\end{equation}
for all $l,j\in\{1,\dots,d\},\; i=1,\dots,Q.$
\end{theorem}

\begin{proof}
    Proof is given in Appendix \ref{Appendix:moment}.
\end{proof}

\begin{corollary}\label{cor:BvM_Q1}
Let \(p=\mathcal{N}(\mu_p,\Sigma_p)\) with positive-definite $\Sigma_p$ and \(h(x)=Ax+b\) (i.e., $Q = 1$).  
Then for any \(q\) with finite second moments,
\[
\mathrm{PSD}_{\mathcal{L}_1}(q\,\|\,p)=0
\;\Longleftrightarrow\;
\left\{
\begin{aligned}
&\mathbb{E}_q[X] = \mu_p,\\
&\mathbb{E}_q[(X-\mu_p)(X-\mu_p)^\top] = \Sigma_p.
\end{aligned}
\right.
\]
\end{corollary}

\begin{proof}
    The result follows directly from the previous theorem by setting \(Q=1\).
\end{proof}

\subsection{Goodness-of-Fit Testing via $\mathrm{PSD}_{\mathcal{L}_1}(q \,\|\, p)$}
\label{section:gof_psd}
The empirical squared discrepancy $\widehat{\mathrm{PSD}}_{\mathcal{L}_1}^2$ can be expressed as a $V$-statistic. Let $\xi(x) \in \mathbb{R}^{Qd^2 + d}$ be a concatenated feature vector consisting of the flattened feature maps $\text{vec}(\Phi_i(x))$ for $i=1,\dots,Q$ and the score $s_p(x)$. We can then write the squared discrepancy as the norm of an empirical mean:
\begin{equation}
\widehat{\mathrm{PSD}}^2_{\mathcal{L}_1}(q \,\|\, p) = \left\| \frac{1}{n} \sum_{j=1}^n \xi(x_j) \right\|_2^2 = \frac{1}{n^2} \sum_{j=1}^n \sum_{k=1}^n k(x_j, x_k),
\end{equation}
where $k(x, x') = \xi(x)^\top \xi(x')$ is a linear kernel in the augmented feature space. This formulation identifies the empirical PSD as a $V$-statistic of rank $Qd^2+d$. By excluding the diagonal terms ($j=k$), one obtains the corresponding $U$-statistic:

\begin{equation}
\widehat{\mathrm{PSD}}_{\mathcal{L}_1, U}^2 = \frac{1}{n(n-1)} \sum_{j \neq k} \xi(x_j)^\top \xi(x_k),
\end{equation}
which provides an unbiased estimator of the squared population discrepancy \citep{liu2016kernelized}. Although this estimator has quadratic computational complexity in n, it can be reformulated as a linear-time expression; see Equation~\ref{eqn:optimal}. This formulation facilitates the use of $U$-statistic asymptotic theory for goodness-of-fit testing \citep{chwialkowski2016kernel, srinivasan2024polynomial}. In practice, implementation typically relies on the wild bootstrap procedure \citep{leucht2013dependent}, which offers robustness in finite samples and handles potential dependencies \citep{chwialkowski2016kernel, srinivasan2024polynomial}. 

However, the statistical power of such tests is  sensitive to the choice of $Q$. Empirical evidence suggests a non-monotonic relationship in certain regimes: while a low order $Q_1$ or a high order $Q_2$ may yield acceptable results, an intermediate choice $Q_1 < Q_3 < Q_2$ can result in significantly degraded performance \citep{srinivasan2024polynomial}. Beyond these local fluctuations, we observe a more fundamental challenge: as the polynomial order $Q$ increases globally, the statistical power consistently degrades (e.g., when $p$ and $q$ differ through moment mismatches of order at most $Q^* < Q$).

To characterise this failure mode, we analyse the asymptotic behaviour of $\mathrm{SNR}^2$ for $\widehat{\mathrm{PSD}}_{\mathcal{L}_1,U}$, focusing on high-order moment contributions under a standard Gaussian target $p = \mathcal{N}(0, I_d)$ and a slightly perturbed sampler $q = \mathcal{N}(\delta \mathbf{1}, \sigma^2 I_d)$.

\begin{theorem}[$\text{SNR}^2$ of $\widehat{\mathrm{PSD}}_{\mathcal{L}_1,U}$]
\label{thm:snr_scaling} Assume $\mathbb{E}_q[\|\xi(X)\|_2^2] < \infty$.
For the $U$-statistic estimator $\widehat{D}_U^2$ of degree $Q$, the $\mathrm{SNR}^2$ satisfies

\begin{equation}
    \mathrm{SNR}(Q)^2 = \frac{\mathbb{E}_q\big[\widehat{D}_U^2\big]^2}{\mathbb{V}_q\big[\widehat{D}_U^2\big]}= \frac{\big(\sum_{i=1}^Q \|\mu_i\|_F^2 + \|\mu_b\|_2^2\big)^2}{\frac{4(n-2)}{n(n-1)} \zeta_{1,Q} + \frac{2}{n(n-1)} \zeta_{2,Q}},
\end{equation}

where $\zeta_{1,Q} = \mathbb{V}_x \big[ \sum_{i=1}^Q \langle \Phi_i(x), \mu_i \rangle_F + s_p(x)^\top \mu_b \big]$ and $\zeta_{2,Q} = \mathbb{E}_{x,x'} \big[ (\xi(x)^\top \xi(x'))^2 \big] - \mathbb{E}_q[\xi(x)^\top \xi(x')]^2$. 

\end{theorem}

\begin{proof}
    This is based on Hoeffding decomposition \citep{hoeffding1992class} of variance of $U$-statistic. See Appendix \ref{Appendix:SNRproof}.
\end{proof}

\begin{remark}[Moment-Dominant Scaling for $\mathrm{SNR}^2$ of $\widehat{\mathrm{PSD}}_{\mathcal{L}_1,U}$ under Gaussianity]
\label{prop:asymptoticsnrboot}
Let $p = \mathcal{N}(0, I_d)$ and $q = \mathcal{N}(\delta \mathbf{1}, \sigma^2 I_d)$ with fixed $(n,d)$. In regimes where $n\zeta_{1,Q} \gg \zeta_{2,Q}$, and under a moment-dominance assumption\footnote{This assumes that under expansion of the numerator and denominator, the highest-order moment terms dominate all lower-order moment contributions in magnitude.},
\begin{align}
    \mathrm{SNR}(Q)^2 \asymp \frac{\big((Q+1)!!\big)^2}{\big((2Q+1)!!\big)}
\end{align}
\end{remark}

\begin{proof}
    Sketch of proof is given in Appendix \ref{Appendix:SNRproof}.
\end{proof}

Remark \ref{prop:asymptoticsnrboot}  provides critical practical implications for both $\widehat{\mathrm{PSD}}_{\mathcal{L}_1,U}$ and $\widehat{\mathrm{PSD}}_{\mathcal{L}_2,U}$. Empirically, although their performance depends on $Q$, simply increasing $Q$ is not effective, as the SNR of the corresponding $U$-statistic deteriorates. In this sense, existing PSD frameworks optimise the \textit{wrong} objective: by focusing solely on maximising the signal, they suffer from uncontrolled variance inflation in higher polynomial orders.  A robust GoF test should instead optimise the SNR (or equivalently $\text{SNR}^2$) instead \cite{sutherland2016generative, jitkrittum2017linear, grathwohl2020learning}.

\subsubsection{$z$-test for $\mathrm{PSD}_{\mathcal{L}_1}$}

In this work, we employ a split-sample testing framework to bypass the $\mathcal{O}(BnQd^2)$ computational overhead associated with bootstrapping high-dimensional $U$-statistics. Given samples $\{x_j\}_{j=1}^n \sim q$, we split the data into training and test sets. The training set is used to estimate the optimal $\{A_1^\star, \cdots, A_Q^\star, b^\star\}$, while the test set is used to construct a one-sample $z$-statistic from $u(x) = \sum_{i=1}^{Q} 
\langle A^*_i, \Phi_i(x) \rangle_F
+
b^{*\top} s_p(x)$ for all $x$ in test set. This construction enables a one-sample two-tailed $z$-test, which circumvents the need for a $B$-iteration bootstrap (see Algorithm \ref{alg:psd_train_test}). The validity of this test follows from the standard central limit theorem.

\begin{proposition}
Let $\{x_j\}_{j=1}^n$ be i.i.d.\ samples and let $\mathcal{D}_{\mathrm{tr}}, \mathcal{D}_{\mathrm{te}}$ be the training and test splits with $|\mathcal{D}_{\mathrm{te}}| \to \infty$. Let $h^\star = h^\star(\mathcal{D}_{\mathrm{tr}})$ denote the test function estimated on the training set. Assume $\mathbb{E}_p[(\mathcal{L}_1 h^\star(X))^2] < \infty$ almost surely. Then under the null hypothesis $H_0: \{x_j\}_{j=1}^n \sim p$, the test statistic $z$ from Algorithm~\ref{alg:psd_train_testzzz} satisfies
\begin{align}
z \;\xrightarrow{d}\; \mathcal{N}(0, 1).
\end{align}
Consequently, the test that rejects $H_0$ when $|z| > z_{1-\alpha/2}$ has asymptotic size $\alpha$.
\end{proposition}

\begin{proof}
The result follows a standard CLT argument, conditionally on $\mathcal{D}_{\mathrm{tr}}$, together with Stein's identity under $H_0$.
\end{proof}

Crucially, this $z$-statistic inherits the same exponential SNR decay as the $U$-statistic.

\begin{remark}
\label{thm:z_snr}
Let $z_Q$ denote the split-sample $z$-statistic constructed using the theoretically optimal coefficients $\{A_i^\star,b^\star\}$ and polynomial degree $Q$. Under the settings of Remark \ref{prop:asymptoticsnrboot} and the moment-dominance argument, the $\text{SNR}^2$ of $z_Q$ deteriorates as $Q \to \infty$. In particular,
\begin{align}
\mathrm{SNR}(Q)^2
\asymp
\frac{\big((Q+1)!!\big)^2}{(2Q+1)!!},
\end{align}
which decays exponentially in $Q$.
\end{remark}

\begin{proof}
    This follows the same arguments as the sketch of proof of Remark \ref{prop:asymptoticsnrboot}.
\end{proof}
This confirms that the degradation of $\text{SNR}^2$ is a fundamental property of the polynomial Stein basis that persists even under computationally efficient split-sample testing.

\paragraph{Practical Considerations for Dependent Sampling}
Under dependent sampling, the test statistic is not necessarily asymptotically standard Gaussian under the CLT, as the variance is typically misspecified and requires correction (e.g., via MCMC covariance estimators; see \citealt{brooks2011handbook}). 

Nevertheless, in our simulation studies we include examples with dependent sampling, and observe that methods derived under the i.i.d.\ assumption can still perform reasonably well in this setting, despite the aforementioned limitations.

\section{$\text{SNR}^2$-Optimised $\mathrm{PSD}_{\mathcal{L}_1}$}
\begin{definition}[$\text{SNR}^2$-optimised $\mathrm{PSD}_{\mathcal{L}_1}$]
The $\text{SNR}^2$-optimised polynomial Stein discrepancy is defined as
\begin{equation}
\label{def:SNROptimised}
\mathrm{o}\text{-}\mathrm{PSD}_{\mathcal{L}_1}(q \,\|\, p)
:=
\sup_{h}
\frac{
\mathbb{E}^2_{q}[\mathcal{L}_1 h(x)]
}{
\mathbb{V}_{q}[\mathcal{L}_1 h(x)]
}.
\end{equation}
Since the objective is scale-invariant, no norm constraint is required.
\end{definition}

Let $\theta \in \mathbb{R}^{Qd^2+d}$ denote the stacked parameters
$\{A_1,\ldots,A_Q,b\}$ and let $\xi(x)$ be defined in Section~\ref{section:gof_psd}, so that
$\mathcal{L}_1 h(x)=\theta^\top \xi(x)$.
Define
\[
\mu=\mathbb{E}_q[\xi(x)],
\qquad
\Sigma=\mathbb{COV}_q[\xi(x)].
\]
Then the objective reduces to the generalised Rayleigh quotient
\begin{equation}
\label{eqn:rayleigh}
J(\theta)
=
\frac{
(\theta^\top \mu)^2
}{
\theta^\top \Sigma \theta
} = \frac{\theta^\top (\mu\mu^\top)\theta}{\theta^\top \Sigma \theta}.
\end{equation}

\begin{proposition}
Given positive-definite $\Sigma$, the objective $J(\boldsymbol{\theta})$ is maximised by $\boldsymbol{\theta}^\star \propto \boldsymbol{\Sigma}^{-1} \boldsymbol{\mu}$. The resulting squared discrepancy $\mathrm{o}\text{-}\mathrm{PSD}_{\mathcal{L}_1}^2(q \,\|\, p)$ is $ \boldsymbol{\mu}^\top \boldsymbol{\Sigma}^{-1} \boldsymbol{\mu}$.
\end{proposition}

\begin{proof}
    Proof is given in Appendix \ref{Appendix:SNRproof}.
\end{proof}

While theoretically optimal, this formulation presents significant challenges. The feature vector dimension is $\mathcal{O}(Qd^2+d)$. In high-dimensional regimes, the empirical covariance $\widehat{\boldsymbol{\Sigma}}$ is rank-deficient, rendering the unregularised estimator non-identifiable. Furthermore, the $\mathcal{O}(Q^3d^6)$ inversion cost is often prohibitive. Thus, without structural regularisation, the naive $\text{SNR}^2$-optimised test is both statistically ill-posed and computationally infeasible. In this work, we do not attempt to compute the estimator for $\boldsymbol{\theta}$ that directly optimises the $\text{SNR}^2$. 

\subsection{Scalable Approximate $\text{SNR}$-Optimised $\text{PSD}_{\mathcal{L}_1}$}

To obtain a tractable approximation, we restrict attention to the span of the unweighted $\mathrm{PSD}_{\mathcal{L}_1}$ solution and introduce scalar weights to reweight each Stein feature.

\begin{proposition}[$\lambda$-PSD$_{\mathcal{L}_1}$]
\label{prop:lambda_psd}
Let $\{A_1^\star,\dots,A_Q^\star,b^\star\}$ denote the theoretical maximiser of
$\mathrm{PSD}_{\mathcal{L}_1}$, and define
\[
\Psi(X):=
\begin{pmatrix}
\langle A_1^\star,\Phi_1(X)\rangle_F \\
\vdots \\
\langle A_Q^\star,\Phi_Q(X)\rangle_F \\
(b^\star)^\top s_p(X)
\end{pmatrix}
\in\mathbb{R}^{Q+1}.
\]
Further define
$\mu_\Psi=\mathbb{E}_q[\Psi(X)]$ and $
\Sigma_\Psi=\mathbb{COV}_q(\Psi(X)).
$

The approximate $\text{SNR}^2$-optimised polynomial Stein discrepancy ($\lambda$-PSD$_{\mathcal{L}_1}$) is
\[
\lambda\mbox{-}\mathrm{PSD}_{\mathcal{L}_1}(q\,\|\,p)
:=
\sup_{\lambda\in\mathbb{R}^{Q+1}}
\frac{\lambda^\top \mu_\Psi}
{\sqrt{\lambda^\top\Sigma_\Psi\lambda}},
\]
whose maximiser satisfies
$
\lambda^\star\propto\Sigma_\Psi^{-1}\mu_\Psi$.
\end{proposition}

Given optimal coefficients ${A_1^\star,\dots,A_Q^\star,b^\star}$, potentially estimated from a sample of size $n_1$, the optimal weight vector $\lambda^\star$ can be estimated in $\mathcal{O}(Q n_2 d^2 + Q^3)$ time using an additional sample of size $n_2$.

\begin{theorem}[SNR Stability of $\lambda\mbox{-}\mathrm{PSD}_{\mathcal{L}_1}$] \label{thm:SNR_weighted_PSD_fisher} Let $p = \mathcal{N}(0, I_d)$ and $q = \mathcal{N}(\delta \mathbf{1}, \sigma^2 I_d)$ with $\|\delta\|_2 > 0$ and $\sigma^2 > 0.5$. Let $\Psi(x) = (\psi_1(x), \dots, \psi_{Q+1}(x))^\top$ be the vector of Stein features spanning the subspace $\mathcal{V}_Q = \mathrm{span}\{\Psi(x)\}$ ($Q \geq 1)$. The $\mathrm{SNR}(Q)^2$ of the optimally weighted statistic $u(x) = \lambda^{\star\top}\Psi(x)$ satisfies: \begin{equation} 0 < \mathrm{SNR}(Q)^2 \leq \chi^2(p\|q). \end{equation} 
\end{theorem}

\begin{proof}
Proof is provided in Appendix~\ref{Appendix:SNRproof}. 
\end{proof}

\subsection{Optimisation of $\lambda$}

In practice, we replace the theoretically optimal coefficients
$\{A_i^\star,b^\star\}$ of $\mathrm{PSD}_{\mathcal{L}_1}$ with empirical estimates
$\{\hat A_i^\star,\hat b^\star\}$. Although the optimal 
$\lambda^\star$ admits the closed-form estimator
$\widehat{\Sigma}_\Psi^{-1}\widehat{\mu}_\Psi$, this approach is often numerically unstable. As $Q$ increases, the empirical covariance matrix $\widehat{\Sigma}_\Psi$ associated with higher-order Stein features may become ill-conditioned, particularly in small-sample regimes. Regularisation is therefore highly desirable. 

In this work, we adopt a simple and computationally efficient stochastic optimisation strategy based on the cross-entropy (CE) method \citep{rubinstein1997optimization,rubinstein2004cross}.  This framework adaptively learns a sampling distribution over \( \lambda \) that concentrates on high-SNR directions. Using a sparse parameterisation, such as a Bernoulli-Gaussian prior (see Algorithm~\ref{alg:bg_cem}), the CE method induces implicit regularisation by pruning components that contribute little to the signal. This approach scales favourably with \( Q \) and improves interpretability by highlighting which moment orders drive the discrepancy.

\subsection{Goodness-of-Fit Testing for $\lambda\mbox{-}\mathrm{PSD}_{\mathcal{L}_1}$}

Given samples $\{x_j\}_{j=1}^n \sim q$, we split the data into training and test sets, where the training set is further divided into two subsets. The first training subset is used to estimate the Stein coefficients $\{A_i^\star\}_{i=1}^Q$ and $b^\star$, while the second is used to estimate the optimal reweighting parameters $\lambda^\star$ given $\{A_i^\star\}_{i=1}^Q$ and $b^\star$. The resulting statistic
$
u(x)
=
\sum_{i=1}^{Q}
\lambda_i^\star
\langle A_i^\star, \Phi_i(x) \rangle_F
+
\lambda_{Q+1}^\star b^{\star\top} s_p(x)
$
is then evaluated on the test set, followed by a one-sample $z$-test; full implementation details are provided in Algorithm~\ref{alg:psd_train_test} in Appendix~\ref{section:algorithms}. Under similar mild conditions, the test statistic is asymptotically Gaussian.

\subsection{Extensions}
The scalar reweighting framework also admits a natural multi-resolution extension. In the most granular setting, one may introduce separate weights for every coefficient in $\{A_i^\star\}_{i=1}^Q$ and $b^\star$, yielding a total of $Qd^2+d$ parameters (equivalent to $\text{o-PSD}_{\mathcal{L}_1}$). More structured alternatives are also possible, such as assigning weights per moment order and dimension, resulting in $Qd+d$ parameters, or employing a hierarchical factorisation with separate weights for moment order and coordinate, requiring only $(Q+1)+d$ parameters. All such formulations can still be written as Rayleigh quotient problems over expanded Stein feature spaces, and Theorem~\ref{thm:SNR_weighted_PSD_fisher} can be extended accordingly. These structured parameterisations may be particularly effective when the discrepancy is concentrated in a small subset of dimensions. The resulting optimisation problems are closely related to sparse generalised eigenvalue literature (see \citealt{tan2018sparse, cai2021note}). We leave these extensions for future work.

\section{Experiments}
\label{section:exp}
This section benchmarks our proposed methods against several well-established Stein discrepancy methods from the literature. The compared methods include:

\begin{itemize}
    \item Gauss KSD: Standard quadratic-time KSD with a Gaussian kernel via the median tune heuristic  \citep{garreau2017large}.
     
    \item IMQ KSD: Standard quadratic-time KSD with the recommended inverse multiquadric (IMQ) kernel using hyperparameters $c = 1$ and $\beta = -0.5$ \citep{gorham2017measuring}.

   \item FSSD-opt: FSSD with a Gaussian kernel and optimised test locations \citep{jitkrittum2017linear}; here, $20\%$ of the samples are used to tune the kernel bandwidth.

   \item RFSD: RFSD with base IMQ kernel and 10 random features \citep{huggins2018random, srinivasan2024polynomial}.
\end{itemize}

Apart from PSD-based methods, the kernel-based methods use bootstrap calibration with $B=500$ simulated null test statistics. Significance level $\alpha=0.05$ is used throughout this section.
Code to reproduce these results is available at \href{https://github.com/olivervu25/weighted-psd}{\texttt{https://github.com/olivervu25/weighted-\\psd}}. This builds on existing implementations from \citealt{jitkrittum2017linear, huggins2018random, srinivasan2024polynomial}. All experiments were run on a single core of a high-performance computing workstation.

\subsection{Experiment 1: Standard Multivariate Gaussian Target}

We consider the target distribution $p$ as
\[
p(x)=\mathcal{N}(0_d,I_d),
\]
and consider the following sampling distributions:

\begin{itemize}
    \item Standard Gaussian: $q_1(x)=\mathcal{N}(0_d,I_d)$. 
    \item Perturbed Gaussian: $q_2(x)=\mathcal{N}(0_d,\Sigma_d)$, where
    \[
    \Sigma_d=\mathrm{Diag}(1.7,1,\ldots,1).
    \]
    \item Student-t: $q_3(x)$: independent Student's $t$ distributions with $5$ degrees of freedom.
\item Laplace: $q_4(x)$ independent $\mathrm{Laplace}(0,1/\sqrt{2})$ distributions; each marginal has the same first and second moments as a standard Gaussian.
\end{itemize}

In Experiment 1, we set \(n=1000\), except for the Student-\(t\) example where \(n=2000\), following \citet{huggins2018random, srinivasan2024polynomial}. We use 200 replicates with significance level \(\alpha=0.05\), and consider dimensions $
d \in \{2,4,8,16,32,64,128\}$. Results are shown in Figures~\ref{fig:popular-methods} and~\ref{fig:psd-vs-wpsd}. Figure~\ref{fig:psd-vs-wpsd} additionally compares \(\mathrm{PSD}_{\mathcal{L}_1}\) and \(\lambda\)-\(\mathrm{PSD}_{\mathcal{L}_1}\) through the estimated mean $\text{SNR}$ (MSNR), computed as the average absolute \(z\)-statistic.

\begin{figure*}[ht]
    \centering
    \includegraphics[width=\linewidth]{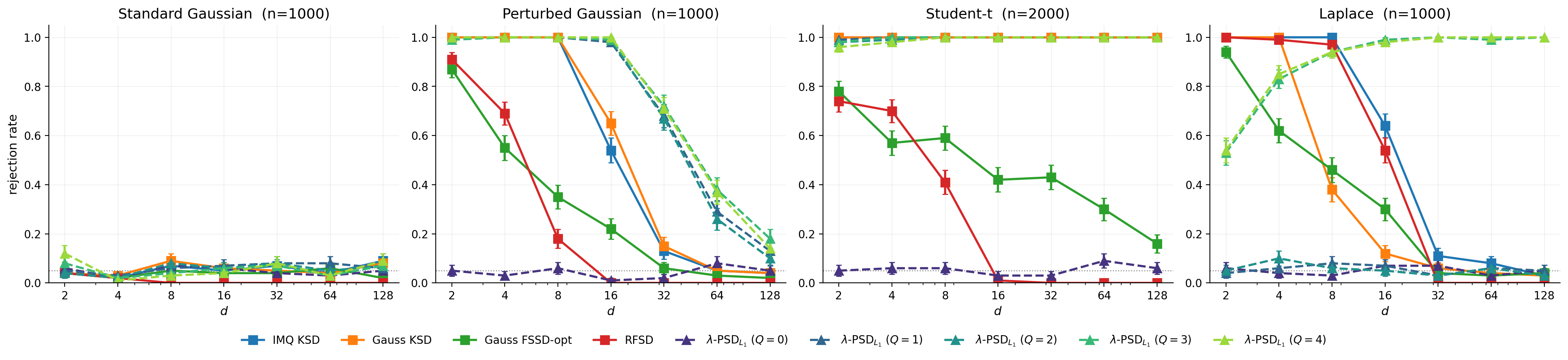}
    \caption{Experiment 1: Rejection rate vs.\ $d$. Dashed/triangles: $\lambda$-$\mathrm{PSD}_{\mathcal{L}_1}$; solid/squares: kernel baselines. Dotted line: $\alpha=0.05$.}
    \label{fig:popular-methods}
\end{figure*}

\begin{figure*}[ht]
    \centering
    \includegraphics[width=\linewidth]{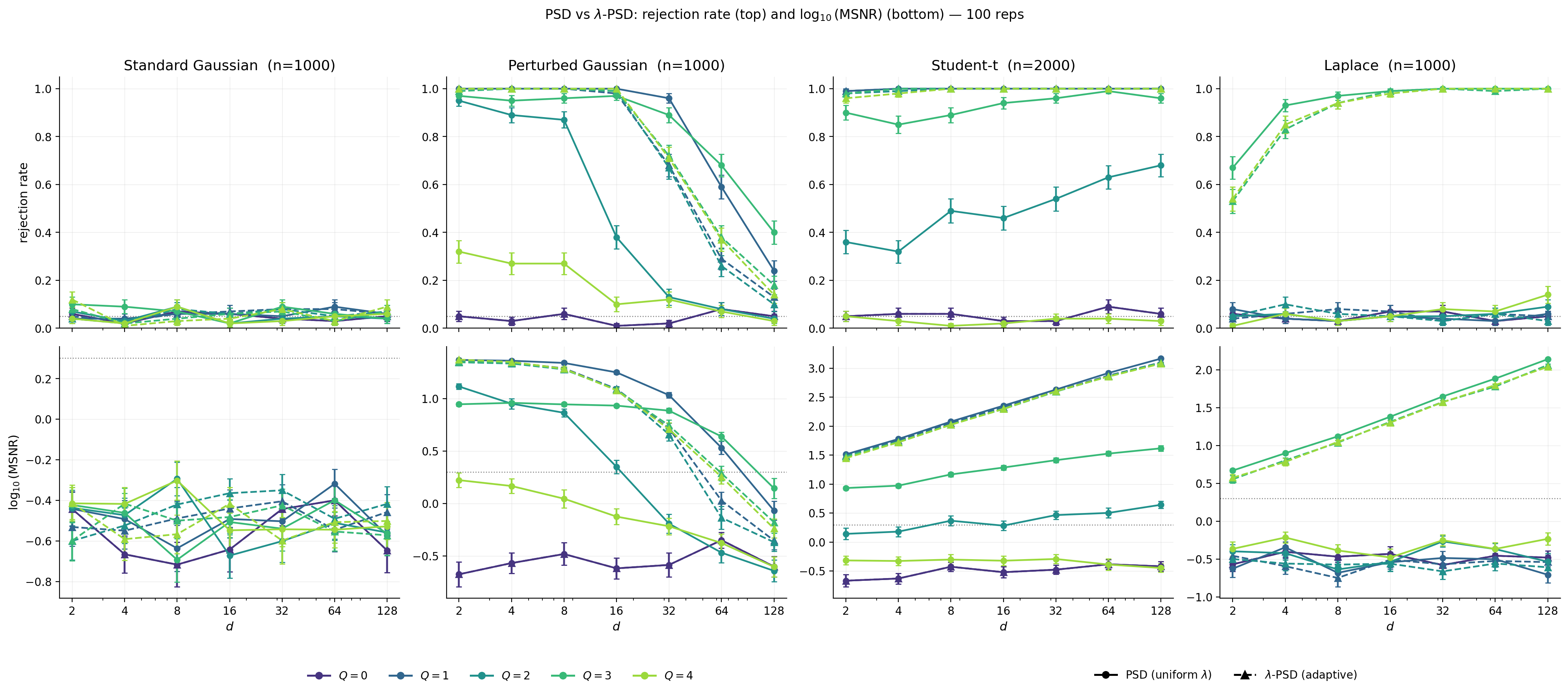}
    \caption{Experiment 1: Rejection rate (top) and $\log_{10}(\mathrm{MSNR})$ (bottom) vs.\ $d$. Solid/circles: $\mathrm{PSD}_{L_1}$; dashed/triangles: $\lambda$-$\mathrm{PSD}_{L_1}$. Colour encodes $Q$. Dotted line: $\log_{10}(1.96)$.}
    \label{fig:psd-vs-wpsd}
\end{figure*}

Figures~\ref{fig:popular-methods} and~\ref{fig:psd-vs-wpsd} show that the proposed $\lambda$-$\mathrm{PSD}_{\mathcal{L}_1}$ achieves strong power across all alternatives while maintaining correct type-I error under the Gaussian null, outperforming IMQ KSD, Gaussian KSD, RFSD, and FSSD-opt in moderate and high dimensions, especially for heavy-tailed targets. Compared with the unweighted $\mathrm{PSD}_{\mathcal{L}_1}$, the adaptive version is more stable across $Q$, avoiding the sensitivity to over-parameterisation seen in the unweighted statistic. For the Laplace example, as expected, low-order choices such as $Q\in\{0,1,2\}$ are insufficient to detect the mismatch, but once $Q$ exceeds the effective complexity of the alternative, both rejection power and $\log_{10}(\mathrm{MSNR})$ remain stably high across $Q$ and dimension $d$.  See Appendix~\ref{Appendix:additional} for additional results demonstrating the superior robustness of $\lambda$-$\mathrm{PSD}_{\mathcal{L}_1}$ with respect to increasing $Q$.

\subsection{Experiment 2: Gaussian-Bernoulli Restricted Boltzmann Machine}

In Experiment 2, following \citealt{liu2016kernelized,jitkrittum2017linear,srinivasan2024polynomial}, we consider a Gaussian-Bernoulli restricted Boltzmann machine (RBM) with non-normalised target density \(p\). The empirical samples  are generated from the same RBM after perturbing the weight matrix \(B \in \mathbb{R}^{50 \times 40}\) with independent Gaussian noise of standard deviation \(\sigma\). When \(\sigma = 0\), the null hypothesis \(H_0: p = q\) holds exactly, whereas for \(\sigma > 0\) the task is to detect deviations induced by the perturbation. We use \(n=1000\) samples generated via Gibbs sampling with 2000 warm-up iterations, and report null rejection rates across a range of perturbation levels based on 200 independent repetitions. The results are presented in Table \ref{tab:rbm_with_wpsd}.

Both $\mathrm{PSD}_{\mathcal{L}_1}$ and $\lambda\mbox{-}\mathrm{PSD}_{\mathcal{L}_1}$ generally outperform the linear-time methods RFSD and FSSD-opt, while achieving performance competitive with the quadratic-time KSD methods. Moreover, $\lambda\mbox{-}\mathrm{PSD}_{\mathcal{L}_1}$ is noticeably less sensitive to the choice of truncation level $Q$ than $\mathrm{PSD}_{\mathcal{L}_1}$, particularly at $\sigma^2=0.02$, where rejection rates remain consistently high across different values of $Q$. These results suggest that our proposed methods may also be effective in dependent sampling and non-Gaussian settings.  
\begin{table}[ht]
\centering
\caption{Null rejection rates at perturbation levels in the RBM example}
\label{tab:rbm_with_wpsd}
\begin{tabular}{l c c c c}
\toprule
Perturbation - $\sigma^2$ & 0.00 & 0.02 & 0.04 & 0.06 \\
\midrule
RFSD & 0.01 & 0.28 & 0.91 & 0.99 \\
FSSD-opt & 0.04 & 0.65 & 0.94 & 0.99 \\
IMQ KSD & 0.04 & 0.98 & 1.00 & 1.00 \\
Gauss KSD & 0.03 & 0.95 & 1.00 & 1.00 \\
\hline
$\text{PSD}_{\mathcal{L}_1}$ ($Q=0$) & 0.07 & 0.23 & 0.84 & 0.96 \\
$\text{PSD}_{\mathcal{L}_1}$ ($Q=1$) & 0.05 & 0.98 & 1.00 & 1.00 \\
$\text{PSD}_{\mathcal{L}_1}$ ($Q=2$) & 0.03 & 0.54 & 0.99 & 1.00 \\
$\text{PSD}_{\mathcal{L}_1}$ ($Q=3$) & 0.04 & 0.97 & 1.00 & 1.00 \\
$\text{PSD}_{\mathcal{L}_1}$ ($Q=4$) & 0.05 & 0.60 & 0.99 & 1.00 \\
\hline
$\lambda$-$\text{PSD}_{\mathcal{L}_1}$ ($Q=1$) & 0.05 & 0.89 & 1.00 & 1.00 \\
$\lambda$-$\text{PSD}_{\mathcal{L}_1}$ ($Q=2$) & 0.04 & 0.84 & 1.00 & 1.00 \\
$\lambda$-$\text{PSD}_{\mathcal{L}_1}$ ($Q=3$) & 0.03 & 0.86 & 1.00 & 1.00 \\
$\lambda$-$\text{PSD}_{\mathcal{L}_1}$ ($Q=4$) & 0.04 & 0.83 & 1.00 & 1.00 \\
\bottomrule
\end{tabular}
\end{table}

\subsection{Runtime Analysis}

We further investigate the computational scalability of sparse $\lambda\mbox{-}\mathrm{PSD}_{\mathcal{L}_1}$ under the standard Gaussian setting. In the first scenario, we fix the dimension at $d=16$ and vary the sample size over $n \in \{100,300,1000,3000,10000\}$. In the second scenario, we fix the sample size at $n=1000$ and vary the dimension over $d \in \{2,4,8,\dots,1024\}$. We consider polynomial orders $Q \in \{0,1,2,3,4\}$ and report average runtime (in seconds) over 100 independent replicates. Results are shown in Figures~\ref{fig:runtime1} and \ref{fig:runtime2}.

\begin{figure}[ht]
    \centering
    \includegraphics[width=0.9\linewidth]{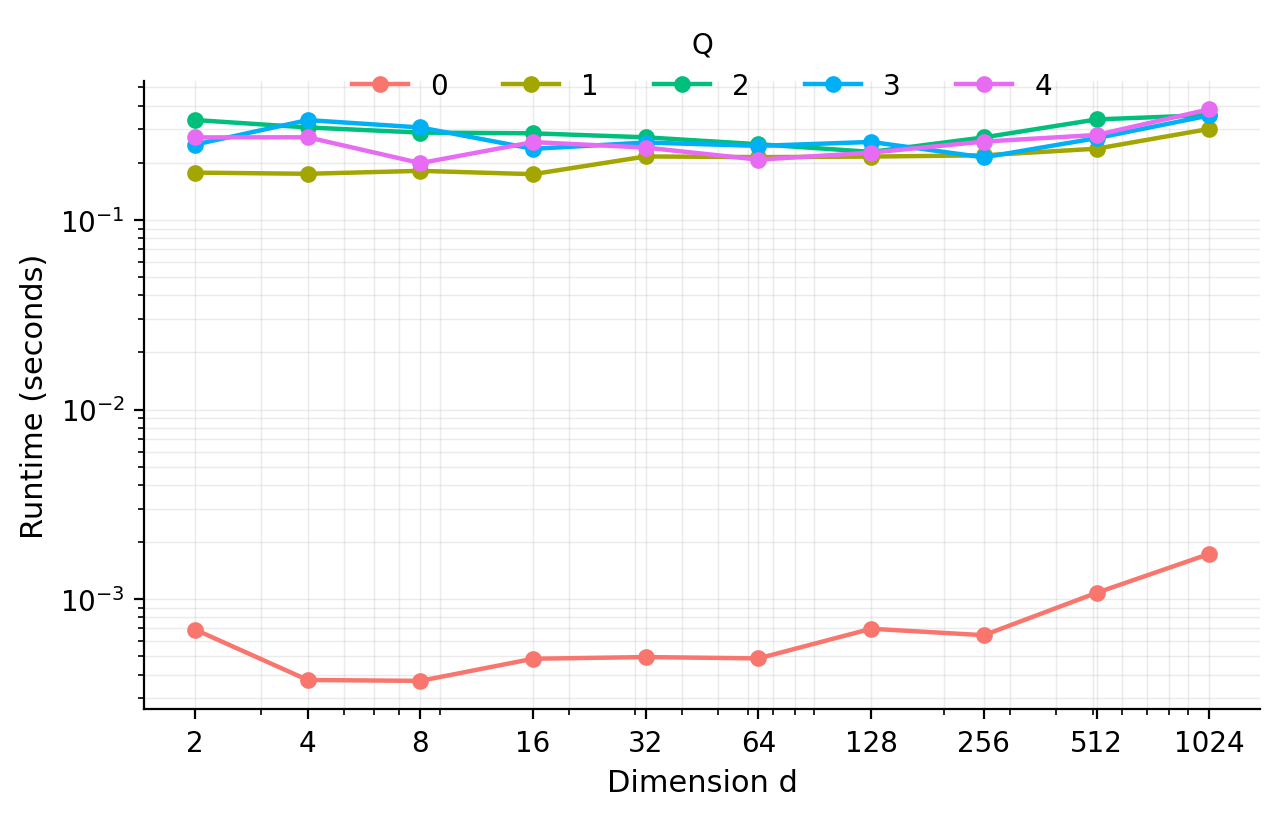}
    \caption{Average runtime of $\lambda\mbox{-}\mathrm{PSD}_{\mathcal{L}_1}$ as a function of dimension $d$ $(n=1000)$. Both axes are shown on logarithmic scales.}
    \label{fig:runtime1}
\end{figure}

\begin{figure}[ht]
    \centering
    \includegraphics[width=0.9\linewidth]{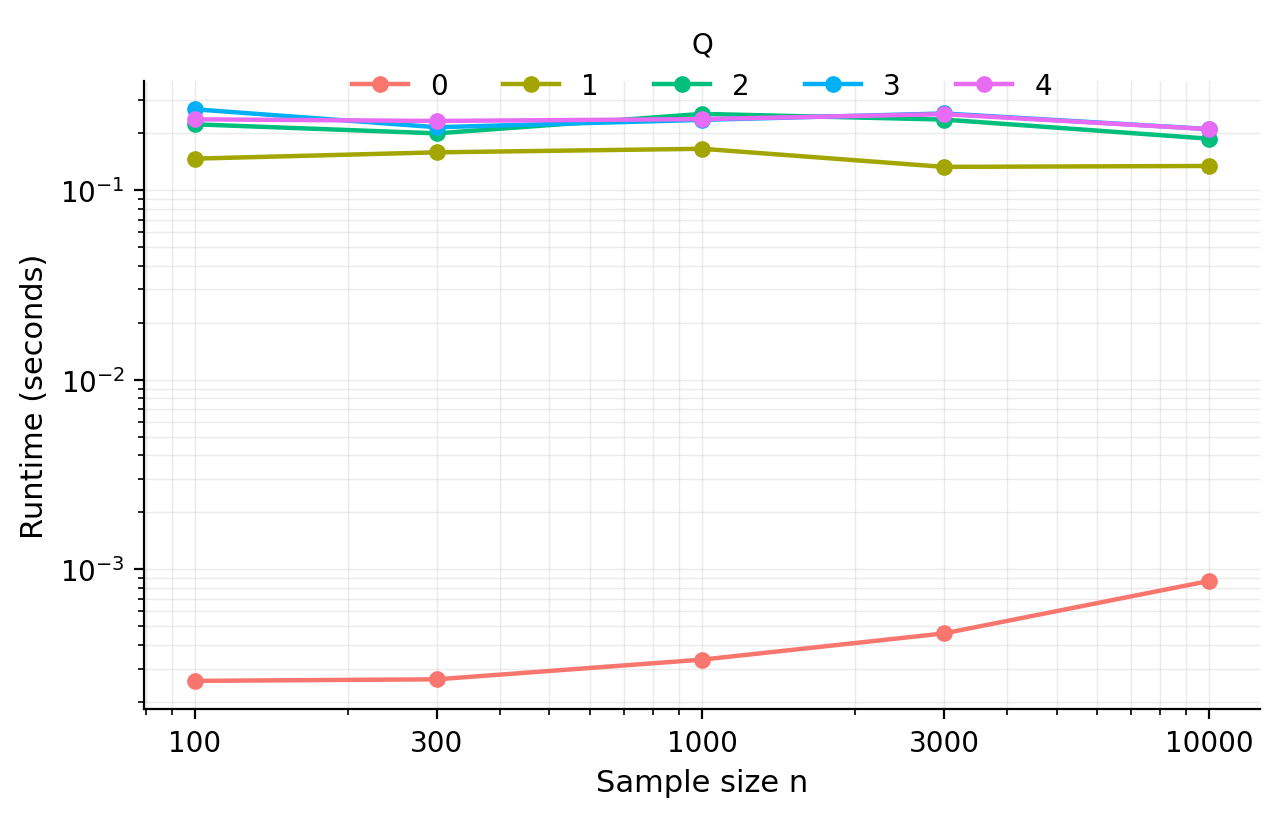}
    \caption{Average runtime of $\lambda\mbox{-}\mathrm{PSD}_{\mathcal{L}_1}$ as a function of sample size $n$ $(d=16)$.  Both axes are shown on logarithmic scales.}
    \label{fig:runtime2}
\end{figure}

As observed in Figures~\ref{fig:runtime1} and \ref{fig:runtime2}, $\lambda\mbox{-}\mathrm{PSD}_{\mathcal{L}_1}$ scales favourably with both $n$ and $d$. Despite the $\mathcal{O}(Qd^2n)$ complexity, the method remains practically fast even for $d=1024$ and $n=10^4$. 

\section{Conclusion}

We analyse the statistical behaviour of  $\mathrm{PSD}_{\mathcal{L}_1}$ from an SNR perspective. Our analysis revealed a previously unstudied failure mode of PSD: although increasing the polynomial order $Q$ enlarges the feature space and increases expressivity, the corresponding goodness-of-fit statistics can suffer severe variance inflation, leading to exponential $\text{SNR}^2$ decay.

Motivated by this observation, we reformulated Stein discrepancy construction as an explicit $\text{SNR}^2$ optimisation problem through a Rayleigh quotient formulation. This led to the proposed $\lambda$-$\mathrm{PSD}_{\mathcal{L}_1}$ framework, which introduces covariance-aware reweighting of Stein features within a low-dimensional subspace. Under Gaussian settings, we showed that $\lambda$-$\mathrm{PSD}_{\mathcal{L}_1}$ stabilises the $\text{SNR}^2$, avoiding the exponential collapse observed in standard PSD. The resulting procedure remains computationally tractable while improving robustness to the choice of polynomial degree. Empirical studies demonstrated that $\lambda$-$\mathrm{PSD}_{\mathcal{L}_1}$ achieves strong goodness-of-fit testing performance across a range of alternatives and dimensions, often outperforming standard kernel Stein methods while retaining linear-time complexity. Like competing approaches, the validity of the associated hypothesis test is asymptotic, and additional adjustments may be required under dependent sampling or small-sample settings.

Beyond the specific construction studied here, our results suggest a broader principle for scalable Stein discrepancy design: optimising the discrepancy alone is insufficient, and effective constructions should directly target the SNR. Future work includes extending the framework to richer reweighting schemes and dependent sampling settings.

\bibliography{example_paper}
\bibliographystyle{icml2026}

%%%%%%%%%%%%%%%%%%%%%%%%%%%%%%%%%%%%%%%%%%%%%%%%%%%%%%%%%%%%%%%%%%%%%%%%%%%%%%%
%%%%%%%%%%%%%%%%%%%%%%%%%%%%%%%%%%%%%%%%%%%%%%%%%%%%%%%%%%%%%%%%%%%%%%%%%%%%%%%
% APPENDIX
%%%%%%%%%%%%%%%%%%%%%%%%%%%%%%%%%%%%%%%%%%%%%%%%%%%%%%%%%%%%%%%%%%%%%%%%%%%%%%%
%%%%%%%%%%%%%%%%%%%%%%%%%%%%%%%%%%%%%%%%%%%%%%%%%%%%%%%%%%%%%%%%%%%%%%%%%%%%%%%
\newpage
\appendix
\onecolumn

\section{Proofs of Moment-Convergence Properties under the BvM Limit}
\label{Appendix:moment}

\textbf{Theorem 3.3. }Let \(p=\mathcal{N}(\mu_p,\Sigma_p)\) with positive-definite $\Sigma_p$, so
\(s_p(x)=-\Sigma^{-1}(x-\mu_p)\). For \(Q\ge1\), define
\begin{equation}
h(x)=\sum_{i=1}^Q A_i\big(x^{\circ i}\big)+b,
\qquad A_i\in\mathbb{R}^{d\times d},\; b\in\mathbb{R}^d.
\end{equation}
Then for any \(q\) with finite moments up to order \(Q+1\),
\begin{equation}
\mathrm{PSD}_{\mathcal{L}_1}(q\,\|\,p)=0
\;\Longleftrightarrow\;
\left\{
\begin{aligned}
&\mathbb{E}_q[X] = \mu_p,\\
&\mathbb{E}_q\!\left[(X_{[l]}-\mu_{p[l]})\,X_{[j]}^{i}\right] 
= \\ &i\,\Sigma_{lj}\,\mathbb{E}_q\!\left[X_{[j]}^{i-1}\right],
\end{aligned}
\right.
\end{equation}
for all $l,j\in\{1,\dots,d\},\; i=1,\dots,Q.$

\begin{proof}
By Proposition~\ref{prop:optimal}, the discrepancy vanishes when every Stein moment vanishes:
\[
\mathrm{PSD}_{\mathcal{L}_1}(q\,\|\,p)=0
\iff
\mu_i:=\mathbb{E}_q[\Phi_i(X)]=0 \;\; \forall i=1,\dots,Q,
\quad
\mu_b:=\mathbb{E}_q[s_p(X)]=0.
\]

Since \(p=\mathcal N(\mu_p,\Sigma_p)\),
\[
s_p(x)=-\Sigma_p^{-1}(x-\mu_p).
\]

We can write these two conditions in turn, exploiting the fact that the Gaussian score is affine map. 

\begin{enumerate}
\item  $ \mu_b=0
\iff
\mathbb E_q[X]=\mu_p.$

Indeed,
\[
\mu_b
=
-\Sigma_p^{-1}\mathbb E_q[X-\mu_p].
\]
As \(\Sigma_p^{-1}\) is invertible,
\[
\mu_b=0
\iff
\mathbb E_q[X-\mu_p]=0
\iff
\mathbb E_q[X]=\mu_p.
\]

\item
For each \(i=1,\dots,Q\), we show that $\mu_i = 0$ is equivalent to a single matrix identity:
\[
\mu_i=0
\iff
\mathbb{E}_q\!\left[(X_{[l]}-\mu_{p[l]})X_{[j]}^i\right]
=
i\,(\Sigma_p)_{lj}\,
\mathbb E_q[X_{[j]}^{i-1}]
\]
for all \(l,j\in\{1,\dots,d\}\).

Recall the definition of the feature map $\Phi_i(x)$:
\[
\Phi_i(x)
=
s_p(x)(x^{\circ i})^\top
+
i\,\mathrm{Diag}(x^{\circ(i-1)}),
\]
so its \((l,j)\)-th entry splits as
\[
[\Phi_i(x)]_{lj}
=
[s_p(x)]_l\,x_{[j]}^i
+
i\,\mathbbm 1(l=j)\,x_{[j]}^{i-1}.
\]

Substituting $[s_p(x)]_l = -\sum_m (\Sigma_p^{-1})_{lm}(x_{[m]}-\mu_{p[m]})$, we have
\[
[\Phi_i(x)]_{lj}
=
-\sum_{m=1}^d
(\Sigma_p^{-1})_{lm}
(x_{[m]}-\mu_{p[m]})x_{[j]}^i
+
i\,\mathbbm 1(l=j)x_{[j]}^{i-1}.
\]

Taking expectation of $\Phi_i(x)]_{lj}$ under $q$ (i.e., $[\mu_i]_{lj}$), it becomes
\[
[\mu_i]_{lj}
=
-\sum_{m=1}^d
(\Sigma_p^{-1})_{lm}
\mathbb E_q[(X_{[m]}-\mu_{p[m]})X_{[j]}^i]
+
i\,\mathbbm 1(l=j)\mathbb E_q[X_{[j]}^{i-1}].
\]

Define the matrix \(M^{(i)}\in\mathbb R^{d\times d}\) where $M^{(i)}_{mj}:=\mathbb E_q[(X_{[m]}-\mu_{p[m]})X_{[j]}^i]$

Then
\[
[\mu_i]_{lj}
=
-(\Sigma_p^{-1}M^{(i)})_{lj}
+
i\,\mathbbm 1(l=j)\mathbb E_q[X_{[j]}^{i-1}].
\]

Since $\mu_i=0$, we have 
\begin{equation*}
\Sigma_p^{-1}M^{(i)}
=
i\,\mathrm{diag}\!\left(
\mathbb E_q[X_{[1]}^{i-1}],
\cdots,
\mathbb E_q[X_{[d]}^{i-1}]
\right) \iff M^{(i)}
=
i\Sigma_p\mathrm{diag}\!\left(
\mathbb E_q[X_{[1]}^{i-1}],
\cdots,
\mathbb E_q[X_{[d]}^{i-1}]
\right),
\end{equation*}

as $\Sigma_p$ is invertible. Taking the \((l,j)\)-th entry yields
\[
\mathbb E_q[(X_{[l]}-\mu_{p[l]})X_{[j]}^i]
=
i\,(\Sigma_p)_{lj}\,
\mathbb E_q[X_{[j]}^{i-1}].
\]
\end{enumerate}
\end{proof}

\section{Proofs of Scaling of $\text{SNR}^2$}
\label{Appendix:SNRproof}

\textbf{Proposition 3.6. }
Let $p = \mathcal{N}(0, I_d)$ and $q = \mathcal{N}(\delta \mathbf{1}, \sigma^2 I_d)$ with fixed $(n,d)$. In regimes where $n\zeta_{1,Q} \gg \zeta_{2,Q}$, and under the moment-dominance assumption,
\begin{align}
    \mathrm{SNR}^2(Q) \propto \frac{\big((Q+1)!!\big)^2}{\big((2Q+1)!!\big)}
\end{align}

\begin{proof}[Sketch of the proof]
For \(p=\mathcal{N}(0,I_d)\), \(s_p(x)=-x\), so each \(\Phi_i(x)\) is a polynomial of degree \(i+1\). Hence
$
\sum_{i=1}^Q \|\mu_i\|_F^2 + \|\mu_b\|_2^2
$
depends on Gaussian moments up to order \(Q+1\), giving growth of order up to \(((Q+1)!!)^2\).

For sufficiently large \(Q\), the dominant term of \(\zeta_{1,Q}\) is
$
\mathbb{V}_q\!\left[
\sum_{i=1}^Q \langle \Phi_i(X),\mu_i\rangle_F
\right].
$
Since \(\Phi_i(X)\) contains degree-\((i+1)\) polynomial features, the variance involves second moments of degree-\((Q+1)\) polynomials, contributing growth of order \((2Q+1)!!\). The coefficients \(\mu_i\) themselves depend on Gaussian moments up to order \(Q+1\), contributing an additional factor of order \(((Q+1)!!)^2\). These give $\zeta_{1,Q}$ the growth of order up to $
(2Q+1)!!\,((Q+1)!!)^2$

For sufficiently large \(n\), the second term in the denominator vanishes asymptotically. As a result,
\begin{equation}
        \mathrm{SNR}^2(Q) \propto \frac{\Big(\big((Q+1)!!\big)^2\Big)^2}{(2Q+1)!!\,((Q+1)!!)^2} = \frac{\big((Q+1)!!\big)^2}{\big((2Q+1)!!\big)}.
\end{equation}
\end{proof}

\textbf{Proposition 4.2. }
The objective $J(\boldsymbol{\theta})$ is maximised by $\boldsymbol{\theta}^\star \propto \boldsymbol{\Sigma}^{-1} \boldsymbol{\mu}$. The resulting squared discrepancy $\mathrm{o}\text{-}\mathrm{PSD}_{\mathcal{L}_1}^2(q \,\|\, p)$ is $ \boldsymbol{\mu}^\top \boldsymbol{\Sigma}^{-1} \boldsymbol{\mu}$.

\begin{proof}
Applying the Cauchy-Schwarz inequality in the $\boldsymbol{\Sigma}$-norm, we have $(\boldsymbol{\theta}^\top \boldsymbol{\mu})^2 = (\boldsymbol{\theta}^\top \boldsymbol{\Sigma} \boldsymbol{\Sigma}^{-1} \boldsymbol{\mu})^2 \leq (\boldsymbol{\theta}^\top \boldsymbol{\Sigma} \boldsymbol{\theta})(\boldsymbol{\mu}^\top \boldsymbol{\Sigma}^{-1} \boldsymbol{\mu})$. Rearranging terms shows $J(\boldsymbol{\theta}) \leq \boldsymbol{\mu}^\top \boldsymbol{\Sigma}^{-1} \boldsymbol{\mu}$, where equality holds if and only if $\boldsymbol{\theta} \propto \boldsymbol{\Sigma}^{-1} \boldsymbol{\mu}$.
\end{proof}

\textbf{Theorem 4.4.}
 Let $p = \mathcal{N}(0, I_d)$ and $q = \mathcal{N}(\delta \mathbf{1}, \sigma^2 I_d)$ with $\|\delta\|_2 > 0$ and $\sigma^2 > 0.5$. Let $\Psi(x) = (\psi_1(x), \dots, \psi_{Q+1}(x))^\top$ be the vector of Stein features spanning the subspace $\mathcal{V}_Q = \mathrm{span}\{\Psi(x)\}$ ($Q \geq 1)$. The $\mathrm{SNR}(Q)^2$ of the optimally weighted statistic $u(x) = \lambda^{\star\top}\Psi(x)$ satisfies: \begin{equation} 0 < \mathrm{SNR}(Q)^2 \leq \chi^2(p\|q). \end{equation}

\begin{proof}
Define the density residual $R(x) = 1 - \frac{p(x)}{q(x)}$; under the specified $q$ and $p$, $\mathbb{E}_q[R(X)] = 0$ and $\mathbb{V}_q[R(X)] = \chi^2(p\|q) < \infty$, i.e., Pearson $\chi^2$-divergence.

For any $\mathbf{f} \in \mathcal{V}_Q$, let $g(x) = \mathcal{L}_{1,p}\mathbf{f}(x)$. Under the Gaussian setting, and $R(x)$ are all square-integrable under $q$. Assuming tail conditions, by Stein's identity, $\mathbb{E}_p[g(X)] = 0$. It is possible to show that:

\begin{align*}
\mathbb{E}_q[g(X)] &= \int g(x) q(x) \, dx \\
&= \int g(x) \left( 1 - \frac{p(x)}{q(x)} + \frac{p(x)}{q(x)} \right) q(x) \, dx \\
&= \int g(x) \left( 1 - \frac{p(x)}{q(x)} \right) q(x) \, dx + \int g(x) p(x) \, dx \\
&= \mathbb{E}_q[g(X) R(X)] + \mathbb{E}_p[g(X)] \\
&= \mathbb{E}_q[g(X) R(X)] + 0 && \text{(by Stein's identity)} \\
&= \mathbb{E}_q[g(X) R(X)] - \mathbb{E}_q[g(X)] \mathbb{E}_q[R(X)] && \text{(since } \mathbb{E}_q[R(X)] = 0 \text{)} \\
&= \mathbb{COV}_q[g(X), R(X)].
\end{align*}

Using this identity, the $\text{SNR}^2$ is expressed as:
\begin{align*}
\mathrm{SNR}(
\textbf{f}
)^2 &= \frac{\mathbb{E}_q[g(X)]^2}{\mathbb{V}_q[g(X)]} \\
&= \frac{\mathbb{COV}_q[g(X), R(X)]^2}{\mathbb{V}_q[g(X)]} \\
&\le \frac{\mathbb{V}_q[g(X)] \mathbb{V}_q[R(X)]}{\mathbb{V}_q[g(X)]} && \text{(by Cauchy-Schwarz)} \\
&= \mathbb{V}_q[R(X)] = \chi^2(p\|q).
\end{align*}

To prove that $\text{SNR}(Q)^2 > 0$, it is sufficient to show that there exists an \textbf{f} $\in \mathcal{V}_Q$ with non-zero SNR. We focus on the optimal solution of $\text{PSD}_{\mathcal{L}_1}$ with $Q = 1$.

Given $p = \mathcal{N}(0, I_d)$, $q = \mathcal{N}(\delta \mathbf{1}, \sigma^2 I_d)$, and the score function $s_p(x) = -x$, the population Stein moments are:
\begin{align*}
\mu_b &= \mathbb{E}_q[-x] = -\delta \mathbf{1} \\
\mu_1 &= \mathbb{E}_q[-XX^\top + I_d] \\
&= I_d - \mathbb{E}_q[XX^\top] \\
&= I_d - \left( \mathbb{COV}_q(X) + \mathbb{E}_q[X]\mathbb{E}_q[X]^\top \right) \\
&= I_d - \left( \sigma^2 I_d + (\delta \mathbf{1})(\delta \mathbf{1})^\top \right) \\
&= I_d - \sigma^2 I_d - \delta^2 \mathbf{1}\mathbf{1}^\top \\
&= (1 - \sigma^2)I_d - \delta^2 \mathbf{1}\mathbf{1}^\top
\end{align*}
The optimal parameters are scaled by $A^\star = \frac{\mu_1}{Z}$ and $b^\star = \frac{\mu_b}{Z}$. Evaluating the operator yields:
\begin{align*}
\mathcal{L}_1 h^\star(x) &= -x^\top A^\star x - x^\top b^\star + \mathrm{Tr}(A^\star) \\
&\propto -x^\top \left[ (1 - \sigma^2)I_d - \delta^2 \mathbf{1}\mathbf{1}^\top \right] x - x^\top (-\delta \mathbf{1}) + \mathrm{Tr}\left( (1 - \sigma^2)I_d - \delta^2 \mathbf{1}\mathbf{1}^\top \right) \\
&= -(1 - \sigma^2)x^\top I_d x + \delta^2 x^\top \mathbf{1}\mathbf{1}^\top x + \delta (\mathbf{1}^\top x) + \mathrm{Tr}\left((1 - \sigma^2)I_d\right) - \mathrm{Tr}\left(\delta^2 \mathbf{1}\mathbf{1}^\top\right) \\
&= (\sigma^2 - 1)\|x\|_2^2 + \delta^2 (\mathbf{1}^\top x)(\mathbf{1}^\top x) + \delta (\mathbf{1}^\top x) + d(1 - \sigma^2) - \delta^2 \mathrm{Tr}\left(\mathbf{1}\mathbf{1}^\top\right) \\
&= (\sigma^2 - 1)\|x\|_2^2 + \delta^2 (\mathbf{1}^\top x)^2 + \delta (\mathbf{1}^\top x) + d(1 - \sigma^2) - d\delta^2 \\
&= \left[ (\sigma^2 - 1)\|x\|_2^2 + \delta^2 (\mathbf{1}^\top x)^2 + \delta (\mathbf{1}^\top x) + d(1 - \sigma^2 - \delta^2) \right]
\end{align*}
Given that 
\begin{align*}
&\mathbb{E}_q\left[\|X\|_2^2\right] = d\sigma^2 + d\delta^2 \\
&\mathbb{E}_q\left[(\mathbf{1}^\top X)^2\right] = 
\mathbb{V}_q(\mathbf{1}^\top X) + \mathbb{E}_q[\mathbf{1}^\top X]^2 = d\sigma^2 + d^2\delta^2 \\ 
&\mathbb{E}_q\left[\mathbf{1}^\top X\right] = d\delta
\end{align*}
Substituting these into the bracketed expression yields:
\begin{align*}
\mathbb{E}_{q,h^*}[g(X)] &= \mathbb{E}_q[\mathcal{L}_1 h^\star(X)] \\ &=\mathbb{E}_q\left[ (\sigma^2 - 1)\|X\|_2^2 + \delta^2 (\mathbf{1}^\top X)^2 + \delta (\mathbf{1}^\top X) + d(1 - \sigma^2 - \delta^2) \right] \\
&= (\sigma^2 - 1)(d\sigma^2 + d\delta^2) + \delta^2(d\sigma^2 + d^2\delta^2) + \delta(d\delta) + d(1 - \sigma^2 - \delta^2) \\
&= d\sigma^4 + d\sigma^2\delta^2 - d\sigma^2 - d\delta^2 + d\sigma^2\delta^2 + d^2\delta^4 + d\delta^2 + d - d\sigma^2 - d\delta^2 \\
&= (d\sigma^4 - 2d\sigma^2 + d) + 2d\sigma^2\delta^2+ d^2\delta^4 - d\delta^2 \\
&= d(\sigma^2 - 1)^2
+ d\delta^2(2\sigma^2 - 1)
+ d^2\delta^4, \\
&\geq d(\sigma^2 - 1)^2
+ d\delta^2(2\sigma^2 - 1)
+ d\delta^4, \\
&\propto  (\sigma^2 - 1)^2
+ \delta^2(2\sigma^2 - 1)
+ \delta^4 \\
&= (\sigma^2 - 1)^2
+ 2\delta^2(\sigma^2 - 1) 
+ \delta^4 + \delta^2 \\
&= (\sigma^2 - 1 + \delta^2)^2 + \delta^2,
\end{align*}
which is non-zero unless $\{\delta, \sigma^2\} = \{0,1\}$. This is also guaranteed by Theorem \ref{thm:BvM_PSD_degreeQ_general}; future versions will explicitly include a lower bound in terms of SNR. Consequently,

\begin{align*}
\mathrm{SNR}(
Q)^2 &= \max_\lambda\frac{\mathbb{E}_q[g(X)]^2}{\mathbb{V}_q[g(X)]} \\
&\geq  \frac{\mathbb{E}_{q,h^*}[g(X)]^2}{\mathbb{V}_{q,h^*}[g(X)]} \\ &> 0 && \text{     (since $\delta \neq 0)$}.
\end{align*}

\end{proof}

\section{Pseudocode of Algorithms}

\label{section:algorithms}

\begin{algorithm}[H]
\caption{PSD$_{\mathcal{L}_1}$ GoF $z$-test}
\label{alg:psd_train_testzzz}
\KwIn{i.i.d., samples $\{x_j\}_{j=1}^n \sim q$, score function $s_p(\cdot)$}
\KwOut{Test statistic and $p$-value}
Randomly split $\{x_j\}_{j=1}^n$ into training set $\mathcal{D}_{\mathrm{tr}}$
and test set $\mathcal{D}_{\mathrm{te}}$ \;

\tcp{Training}
Compute empirical Stein moments $\{\hat{\mu}_i\}, \hat{\mu}_b$ on $\mathcal{D}_{\mathrm{tr}}$\;
Compute optimal maximiser $\{A_1^\star, \cdots,A_q^\star, b^\star\}$ in closed form (as in Proposition \eqref{prop:optimal})\;

\tcp{Testing}
Evaluate $u_j = \sum_{i=1}^{Q}
\langle A^*_i, \Phi_i(x) \rangle_F
+
b^{*\top} s_p(x)$ for all $x_j\in\mathcal{D}_{\mathrm{te}}$\;
Compute the one-sample $z$-statistic
\begin{equation}
z = \frac{\sqrt{|\mathcal{D}_{\mathrm{te}}|}\,\bar{u}}{s_u}
\end{equation}

where $\bar{u}=|\mathcal{D}_{\mathrm{te}}|^{-1}\sum_j u_j$
and $s_u^2=(|\mathcal{D}_{\mathrm{te}}|-1)^{-1}\sum_j (u_j-\bar{u})^2$\;
Return the $z$-statistic and the corresponding two-sided $p$-value\;

\end{algorithm}

\begin{algorithm}
\footnotesize
\setlength{\baselineskip}{0.9\baselineskip}
\caption{Bernoulli-Gaussian Cross Entropy Method (BG-CEM) for $\lambda$ optimisation}
\label{alg:bg_cem}
\KwIn{Feature matrix $G \in \mathbb{R}^{n \times (Q+1)}$ (see Algorithm \ref{alg:psd_train_test}), maximum number of iterations $T = 100$, candidates $M = 1000$, elite size $K = 10$, tolerance $\epsilon = 10^{-3}$, patience $p = 3$}
\KwOut{Optimised weights $\hat{\lambda}$, sign (i.e., $\pm 1$)}

$n \leftarrow \mathrm{nrow}(G)$\;

Precompute sample mean and sample covariance matrix of $G$
\[
\widehat{\mu_G} \leftarrow \frac{1}{n} G^\top \mathbf{1}_n,
\qquad
\widehat{\Sigma_G} \leftarrow \frac{1}{n} G^\top \left(I_n - \frac{1}{n}\mathbf{1}_n \mathbf{1}_n^\top \right) G
\]

Initialise activation probabilities $p_k = 0.5$, means $\mu_k = 0$, and scales $\sigma_k = 1$ for $k=1,\dots,Q+1$\;

Set best score $s^\star \leftarrow -\infty$, best $\lambda^* \leftarrow \text{null}$\;

\For{$t = 1,\dots,T$}{

\For{$m = 1,\dots,M$}{

Sample $z_k^{(m)} \sim \mathrm{Bernoulli}(p_k)$ and $\epsilon_k^{(m)} \sim \mathcal{N}(\mu_k,\sigma_k^2)$ independently\;

Form candidate
\[
\lambda^{(m)}_k = z_k^{(m)} \cdot \epsilon_k^{(m)}, \quad k=1,\dots,Q+1
\]
and normalise $\lambda^{(m)} \leftarrow \lambda^{(m)} / \|\lambda^{(m)}\|$\;

Compute objective using precomputed moments
\[
s^{(m)} = \frac{(\widehat{\mu_G}^\top \lambda^{(m)})^2}{\lambda^{(m)\top} \widehat{\Sigma_G} \lambda^{(m)}}
\]

Update best solution $\lambda^*$ if $s^{(m)} > s^\star$\;

}

Select elite set $\mathcal{A}$ of top $K$ candidates according to $s^{(m)}$\;

Early stop if score improvement is less than $\epsilon$ for $p$ iterations\;

Update parameters:
\[
p_k \leftarrow \frac{1}{|\mathcal{A}|} \sum_{\lambda \in \mathcal{A}} \mathbf{1}\{|\lambda_k| > 0\},
\quad \mu_k \leftarrow \mathrm{mean}\big(\{\lambda_k : \lambda \in \mathcal{A}\}\big), 
\quad
\sigma_k \leftarrow \mathrm{sd}\big(\{\lambda_k : \lambda \in \mathcal{A}\}\big)
\]

}

Return $\{\lambda^*, \text{sign}(\widehat{\mu_G}\lambda^*)\}$
\end{algorithm}

\begin{algorithm}[H]

\footnotesize
\setlength{\baselineskip}{0.9\baselineskip}

\caption{$\lambda$-PSD$_{\mathcal{L}_1}$ GoF test}
\label{alg:psd_train_test}
\KwIn{Samples $\{x_j\}_{j=1}^n \sim q$, score function $s_p(\cdot)$}
\KwOut{Test statistic and $p$-value}

Randomly split $\{x_j\}_{j=1}^n$ into training set $\mathcal{D}_{\mathrm{tr}}$
and test set $\mathcal{D}_{\mathrm{te}}$\;

Further split $\mathcal{D}_{\mathrm{tr}}$ into $\mathcal{D}_1$ and $\mathcal{D}_2$\;

\tcp{Training (Stage 1)}

Estimate coefficients $\{\hat{A}^*_i\}_{i=1}^Q$ and $\hat{b}^*$ using $\mathcal{D}_1$\;

\tcp{Training (Stage 2)}
Initialise $G \in \mathbb{R}^{n_2 \times (Q+1)}$\;
\For{$k \leftarrow 1$ \KwTo $n_2$}{
    \For{$i \leftarrow 1$ \KwTo $Q$}{
        $G_{[k,i]} \leftarrow \langle \hat{A}_i^\star, \Phi_i(x_k) \rangle_F$\;
    }
    $G_{[k,Q+1]} \leftarrow (\hat{b}^\star)^\top s_p(x_k)$\;
}

Get $\{\lambda^*_1, \cdots,\lambda^*_{Q+1}\} $ from $\text{BG-CEM}(G)$ \\

\tcp{Testing}
Evaluate $u_j = \sum_{i=1}^Q \lambda_i^* \langle \hat{A}_i^\star, \Phi_i(x)\rangle_F
+
\lambda^*_{Q+1} (\hat{b}_i^\star)^\top s_p(x)$
for all $x_j \in \mathcal{D}_{\mathrm{te}}$\;

Compute the one-sample $z$-statistic
\begin{equation}
z = \frac{\sqrt{|\mathcal{D}_{\mathrm{te}}|}\,\bar{u}}{s_u}
\end{equation}

where $\bar{u}=|\mathcal{D}_{\mathrm{te}}|^{-1}\sum_j u_j$
and $s_u^2=(|\mathcal{D}_{\mathrm{te}}|-1)^{-1}\sum_j (u_j-\bar{u})^2$\;

Return the $z$-statistic and the corresponding two-sided $p$-value\;

\end{algorithm}

\section{Additional Results}
\label{Appendix:additional}

We additionally report supplementary results for Experiment~1 in which the dimension is fixed at $d=16$, while the sample size remains as specified in Section~\ref{section:exp}. Here, we vary the polynomial order over $Q \in \{0,1,\ldots,8\}$. As shown in Figure~\ref{fig:msnr_rejection_grid}, $\lambda$-$\mathrm{PSD}_{\mathcal{L}_1}$ maintains relatively stable rejection behaviour and $\log_{10}(\mathrm{MSNR})$ as $Q$ increases, whereas $\mathrm{PSD}_{\mathcal{L}_1}$ exhibits substantially larger fluctuations across polynomial orders. These results further support the superior robustness of $\lambda$-$\mathrm{PSD}_{\mathcal{L}_1}$ with respect to  $Q$.

\begin{figure}[H]
\centering
\includegraphics[width=\textwidth]{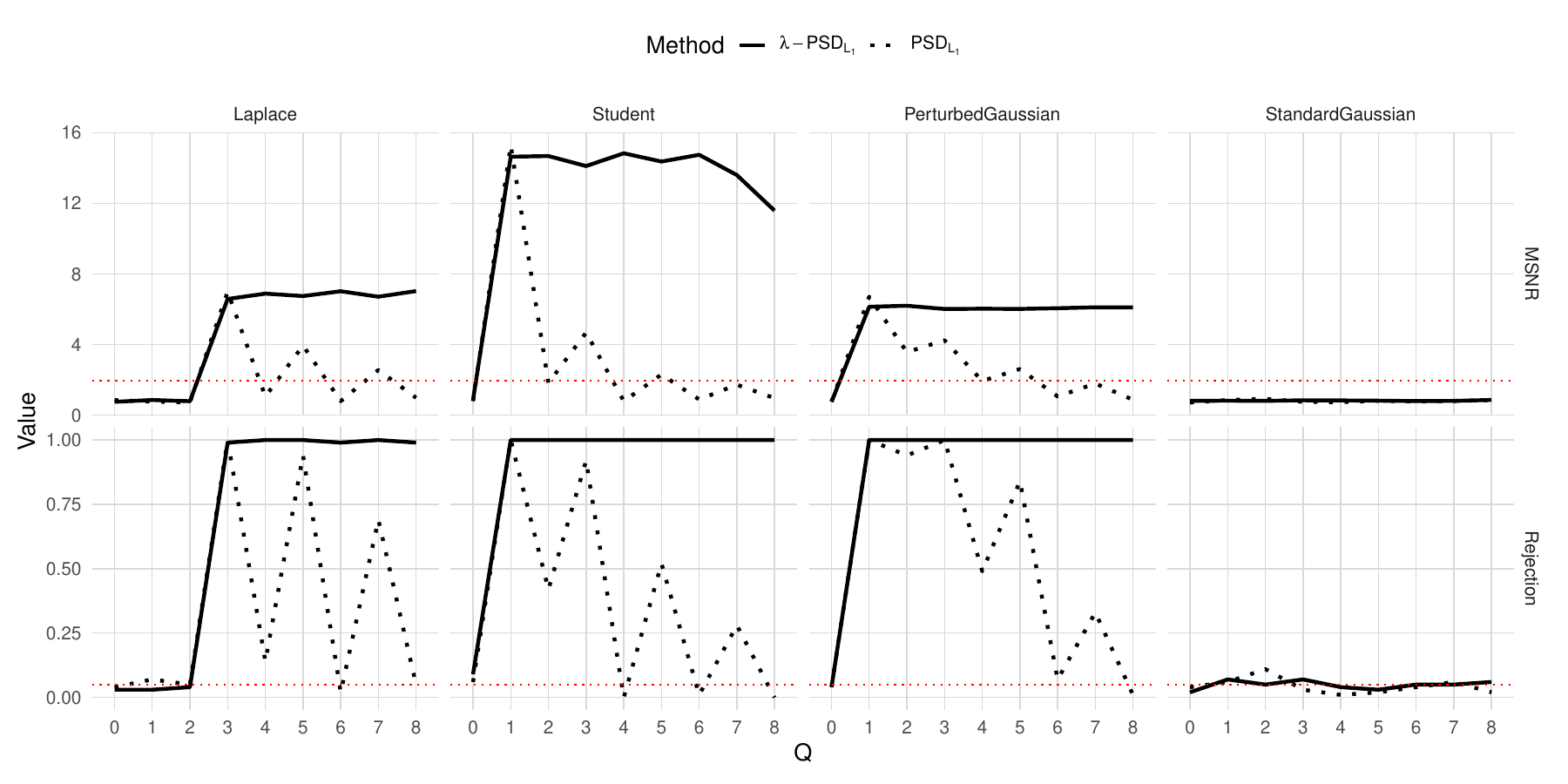}
\caption{Experiment 1: Comparison of rejection rates and $\log_{10}(\mathrm{MSNR})$ between $\mathrm{PSD}_{\mathcal{L}_1}$ and $\lambda$-$\mathrm{PSD}_{\mathcal{L}_1}$ for varying polynomial orders $Q \in \{0,1,\ldots,8\}$ at fixed dimension $d=16$.}
\label{fig:msnr_rejection_grid}
\end{figure}

\end{document}